\documentclass[lettersize,journal]{IEEEtran}
\usepackage{graphicx}    % 插入图片
\usepackage{subcaption}  % 支持 subfigure
\usepackage{amsmath,amsfonts,amssymb}
\usepackage{algorithmic}
\usepackage{algorithm}
\usepackage{array}
\usepackage{tabularx, booktabs}
\usepackage{textcomp}
\usepackage{stfloats}
\usepackage{url}
\usepackage{verbatim}
\usepackage{graphicx}
\usepackage{cite}
\usepackage{multirow} % 导言区添加
\usepackage{tikz}

\newtheorem{definition}{\bf Definition}
\newtheorem{proposition}{\bf Proposition}
\newtheorem{theorem}{\bf Theorem}
\newtheorem{lemma}{\bf Lemma}
\newtheorem{corollary}{\bf Corollary}
\newtheorem{remark}{\bf Remark}

\usetikzlibrary{automata, positioning, arrows.meta}
\begin{document}

\title{Capacity Scalability of LEO Constellations With Dynamic Link Failures}

\author{
  \IEEEauthorblockN{Wei Li\IEEEauthorrefmark
{1}, \:\: Min Sheng\IEEEauthorrefmark
{1}, \emph{Fellow, IEEE}}
%,  Pasquale Pace\IEEEauthorrefmark
%{2} \emph{Member, IEEE}, \\
%Giancarlo Fortino\IEEEauthorrefmark
%{2},  \emph{Fellow, IEEE},  and Jiandong Li\IEEEauthorrefmark
%{1},  \emph{Fellow, IEEE}
\\
\IEEEauthorblockA{ \IEEEauthorrefmark
{1}State Key Laboratory of Integrated Service Networks, Xidian University, Xi’an, Shaanxi, 710071, China}\\
Email: xidianliwei@stu.xidian.edu.cn, msheng@mail.xidian.edu.cn\\
}



\maketitle

\begin{abstract} Dynamic link failures disrupt the connectivity and geometric symmetry of the constellation structure, thereby increasing protocol overhead and degrading the effective capacity for traffic transport. The fundamental relationship between constellation size and effective capacity under protocol overhead constraints  remains unclear.  To this end, we define capacity scalability as the ratio of constellation capacity under non-failure conditions to protocol overhead.   Specifically, if ISL states follow a two-state discrete Markov chain and the maintenance period is $k \geq 1$, the upper bound of capacity scalability under the uniform traffic pattern is $O(1/n)$, where $n$ is the number of satellites. With perfect information about the constellation topology, the upper bound  can be achieved via shortest-path routing.  For any given protocol, there exists an optimal constellation deployment scale in terms of capacity scalability. When the constellation size is below this optimum scale, capacity scalability increases with constellation size, thereby improving effective capacity. Increasing the maintenance period $k$ can improve capacity scalability, but it does not change the fact that the capacity scalability  converges to zero when the constellation size exceeds the optimal scale.
 
\end{abstract}


\section{Introduction}
\IEEEPARstart{A}{s} more satellites are deployed in low Earth orbit, the available orbital space resources have decreased rapidly, leading to an exponential increase in the frequency of collision avoidance maneuvers \cite{letizia2017extending,reiland2021assessing}. Frequent maneuvers and unpredictable space weather (e.g., solar radiation, micro-space debris  and single-event upsets) inevitably cause variations in inter-satellite link (ISL) states, thereby disrupting the connectivity and symmetry of the constellation \cite{olivieri2020large,sheng2025effects,ramanathan2024weathering}.  However, the data transport capacity of the constellation is closely dependent on the stability of the constellation structure and the efficiency of the protocols\cite{11129665,sheng2023coverage}. The main function of the protocols is to maintain ISLs with low overhead and to quickly plan multi-hop routing paths for data packets in the dynamic constellation structure.  The operation of the protocols essentially consumes the capacity provided by the dynamic constellation structure. Therefore, there exists a fundamental trade-off between protocol overhead and constellation capacity, especially in large-scale constellations. This directly determines whether the capacity gain from increasing the number of satellites can be effectively used for transporting data packets in a constellation with dynamic link failures.

Existing studies on capacity and protocol overhead are usually done separately. Research on constellation capacity mainly focuses on improving capacity and analyzing its influencing factors, such as traffic distribution\cite{1642730,sun2003capacity,6465577,ramanathan2017symptotics}, ISL failures\cite{10230313,10361587,10017632} and ISL establishment rules\cite{11481147,10826930,10622303,11044684}. In contrast, research on protocol overhead mainly focuses on designing specific strategies, such as distributed or centralized routing\cite{ramakanth2025,11068466,11068516,11027048,10356345}, load balancing strategies\cite{li2024analyzing,li2024robustness,BAI2025111089,ramakanth2026optimal} and ISL maintenance strategies\cite{11148805,11161010}. This makes existing results unable to explain how protocol overhead constrains capacity as the constellation scale increases, as well as how ISL evolution behavior affects protocol overhead.  To fill this gap, we define capacity scalability as the ratio of constellation capacity under non-failure conditions to protocol overhead. Based on this definition, we  analyze the fundamental behavior of capacity scalability under dynamic link failures. The total overhead includes contention overhead ($\sigma$), incurred during data packet routing and determined by the number of hops and the number of source–destination pairs, as well as consensus overhead ($\omega$), incurred for ISL maintenance and determined by ISL evolution the  maintenance period and the number of ISLs. Overall,  the main contributions are summarized as follows:

\begin{itemize}
   
    \item We use a two-state discrete-time Markov chain to describe the dynamic evolution and memory behavior of ISLs. The lower bound of the consensus overhead required to maintain each ISL is derived in terms of information entropy. Based on this result, the upper bound of capacity scalability under the uniform traffic model is obtained as
    \begin{equation}
   \tau\leq \tau(n,k) = \frac{16\sqrt{n}}{1 + \sigma n^{1.5} + 4n\bar{H}_k(\alpha,\beta)}=O(1/n),
    \label{ITR1}
\end{equation}
where $\alpha$ and $\beta$ are the failure probability and the recovery probability, respectively. $\sigma$ denotes the overhead incurred for transmitting a data packet over each hop. $\bar{H}_k(\alpha,\beta)
= k^{-1} \left[ h(\pi_{1}) + (k-1)\bigl(\pi_{1} h(\alpha) + \pi_{0} h(\beta)\bigr) \right]$ is  the lower bound of the consensus overhead required to maintain each ISL every $k$ time slots, $h(x) = -x\log_2x - (1-x)\log_2(1-x),\:x>0$. $\pi_{1}$ and $\pi_{0}$ are the steady-state probabilities of the ISL evolution. The overhead grows much faster, at a rate of $\Theta(n^{1.5})$, than the capacity growth rate of $\Theta(\sqrt{n})$\footnote{$f(n)=O(g(n))$ if there exist $c>0$ and $n_{0}>0$ such
that $f(n)<cg(n)$ for $n>n_{0}$. $f(n)=\Omega(g(n))$ if $g(n)=O(f(n))$.
$f(n)=\Theta(g(n))$ if $f(n)=O(g(n))$ and $g(n)=O(f(n))$.}. This indicates that there exists an optimal constellation size $n^{*}$ that maximizes capacity scalability.  $n^{*}$ is determined by solving the equation $2\sigma n^{1.5} + 4n\bar{H}_k(\alpha,\beta) = 1$.  When $n < n^{*}$, the capacity scalability increases as $n$ increases, so increasing the constellation size is beneficial. In contrast, when $n > n^{*}$, the capacity scalability converges to zero as $n$ increases, so the gain from further increasing the constellation size converges to zero. As the contention overhead (consensus overhead) approaches zero, the optimal constellation size is given by $n_{*} \approx \bar{H}_k^{-1}(\alpha,\beta)$ ($\left( 2\sigma \right)^{-2/3}$). The results further indicate that, for the same reduction in overhead, the optimal constellation size is more sensitive to consensus overhead than to contention overhead.  In particular, when the total contention overhead is constant ($\sigma = c$), the upper bound of capacity scalability increases from $O(1/n)$ to $O(1/\sqrt{n})$, which merely  reduces the convergence rate of capacity scalability to zero.  To maintain the capacity scalability, therefore, it is necessary to reduce the number of source-destination pairs, the average hop count and the consensus overhead. 
    
    \item For any ISL failure evolution behavior, increasing the ISL maintenance period ($k$) can only improve capacity scalability but cannot change the inherent trend that capacity scalability converges to zero when $n > n^{*}$. This is because the lower bound of the consensus overhead required to maintain each ISL is non-zero, i.e., 
    $$\omega \geq \bar{H}_{\infty}(\alpha, \beta) =  \frac{\beta}{\alpha+\beta} h(\alpha) + \frac{\alpha}{\alpha+\beta} h(\beta)>0.$$
Furthermore, the capacity scalability is analyzed under different ISL failure behaviors. When the ISL state exhibits stronger memory, i.e., $|\mu| = |1 - \alpha - \beta|$ is large, increasing the maintenance period $k$ yields little improvement in capacity scalability. This is due to the high predictability of ISL evolution behavior. If an ISL is in a failed (non-failed) state at time slot $t$, it remains in the failed (non-failed) state at time slot $t+k$ with high probability. We also observe that optimal capacity scalability is achieved when the ISL failure probability is low while the recovery probability is high. Simulation experiments are performed to validate the theoretical observations.   
\end{itemize}

The rest of this paper is organized as follows.   The related work on existing studies is reviewed in Section \ref{RWII}. The system model and ISL failure model are presented in Section \ref{SMIII}.   In Section \ref{SECII}, the lower bound of consensus overhead, the upper bound of capacity scalability and the optimal constellation size in terms of scalability are derived.  We discuss the impact of the maintenance period and ISL evolution behavior on the capacity scaling law separately in Section \ref{SECV}. Based on the shortest-path routing strategy, in Section \ref{SECVI}, we simulate the ISL evolution behavior and verify the upper bound of capacity scalability.

\section{RELATED WORKS\label{RWII}}

A LEO constellation with a one-satellite, four-link configuration is equivalent to a 2-Torus structure, except for polar constellations\cite{1642730,sun2003capacity,11044684}. Under uniform traffic pattern, the capacity upper bound of a symmetric 2D-torus structure is $8\sqrt{n}$, by applying bisection arguments \cite{1642730}. 
This result shows that the capacity growth rate achieved by increasing the constellation size is at most $O(1/\sqrt{n})$. Under this framework, the lower bound on spare capacity for recovering from a link or node failure is derived, together with routing and restoration schemes that achieve this bound \cite{sun2003capacity}. The results show that the spare capacity requirement for a link-based restoration scheme is nearly $n$ times that for a path-based scheme. In \cite{11481147,10826930,10622303,11044684}, the effects of ISL establishment rules and ISL number on constellation performance under uniform traffic are studied, aiming to achieve the capacity upper bound under failure cases. In \cite{11481147}, a dynamic ISL reconfiguration method is proposed to mitigate failure effects by reinforcing bottleneck cuts. In \cite{10826930}, a topology design method is proposed that maximizes capacity by twisting edge ISLs in a 2D-torus structure. In \cite{10622303}, a topology design algorithm is proposed to minimize redundant ISLs caused by constellation asymmetry while maintaining constellation capacity.  These works focus on specific ISL establishment rules to maintain or improve constellation capacity and lack generality. To address this theoretical limitation, the lower bounds of the average shortest path length in vertex-symmetric constellations are derived, and it is shown that alternative topologies can outperform the standard mesh grid in both diameter and efficiency \cite{11044684}.  In \cite{10230313,10361587,10017632}, by modeling complex dynamic cascading failure processes of nodes and ISLs in LEO constellations, the impacts of buffer limits, processing constraints, and real-time routing on network-wide outages are analyzed through simulations.  

Although these studies provide theoretical results on capacity and valuable insights, they all neglect the  overhead required to implement these algorithms, which is especially important in large-scale constellations. Fortunately, the effective capacity incorporating protocol overhead has been explored in large-scale grid-based wireless ad hoc networks.  A framework is developed to evaluate the capacity and scalability of finite grid-based wireless networks under different node degrees. It captures protocol overhead, congestion bottlenecks, traffic heterogeneity and other real-world concerns \cite{6465577,ramanathan2017symptotics}.  These studies on wireless ad hoc networks provide useful insights for analyzing the capacity scalability of LEO constellations under dynamic link failures. In addition to routing overhead, the main difference is that LEO constellations focus on ISL maintenance overhead rather than interference avoidance overhead. Therefore, these results cannot be directly used to explain capacity scalability in LEO constellations.

The consensus among the aforementioned theoretical studies on constellation capacity is that the design of intelligent routing and  ISL establishment algorithms plays a key role in enhancing and maintaining the constellation capacity. Although routing design in wireless ad hoc networks is not a new topic, recent studies have extensively investigated routing design and evaluation in LEO constellations from multiple perspectives.  In \cite{ramakanth2025},  the trade-off between centralized and distributed routing in faulty constellations is studied. It shows that in highly dynamic bufferless networks, the distributed scheme achieves higher throughput than the centralized scheme, while in buffered networks it yields lower delay. To reduce packet loss and improve path reliability, distributed routing algorithms have been designed based on the predictable geometric structure of the constellation, achieving significant improvements in path success ratio, delay and average hop count\cite{11068466,11068516,11027048,10356345}.  ISL establishment rules \cite{10356345}, satellite-to-ground reliable communication \cite{11068516} and hybrid centralized/distributed strategies \cite{11068516,11027048,10356345} are also considered in the routing design. Essentially, these studies utilize the distributed routing design paradigm to improve adaptability to dynamic topologies.
Due to the limited transport capacity of individual satellites and ISLs, after partial route reconfiguration, packet congestion is highly likely to occur at certain satellites or ISLs, thereby increasing end-to-end latency and packet loss rate \cite{li2024analyzing,li2024robustness}. This phenomenon is referred to as cascading failures. The spatial non-uniformity of traffic carried by the LEO constellation, resulting from differences in population density, further intensifies these cascading effects. To suppress cascading failures and balance traffic load, a virtual node model is used for load balancing in \cite{li2024analyzing} by integrating ground station deployment into the network model. In contrast, cascading failures are analyzed by combining complex network and hypernetwork theories in \cite{li2024robustness}, where corresponding strategies are developed for different types of failures and attack scenarios. An intelligent elastic routing framework based on graph neural networks and deep reinforcement learning is also proposed in \cite{BAI2025111089}. The load balancing problem under sparse traffic is modeled as a linear programming problem, establishing the lower bound on the worst-case link load for any oblivious routing scheme and presenting a scheme that achieves this bound \cite{ramakanth2026optimal}. The effects of unpredictable space weather (such as single-event upsets and solar storms) on the fundamental performance of satellite constellations are also assessed in \cite{11148805,11161010}. These studies offer important insights into the performance of satellite constellations under failures. However, they generally neglect the relationship between protocol overhead and capacity as the constellation size increases.

\section{System Models\label{SMIII}}
A LEO constellation consists of $P$ orbital planes, with $M$ satellites evenly distributed in each orbital plane. The position of each satellite is represented as $(x, y)$, where $x \in \{0, 1, \dots, P-1\}$ denotes the index of the orbital plane, and $y \in \{0, 1, \dots, M-1\}$ denotes the index of the satellite within orbital plane $x$. Each satellite establishes two co-plane links and two inter-plane links. In the ideal case, the number of operational satellites  is
$n=MP$. For describing the geometric  of the LEO constellation, it is  represented as \( \theta: n/P/F \). \( \theta \in (0^\circ, 90^\circ) \) is the inclination angle of the orbital planes with respect to the equator, determining the areas on Earth covered by the satellites. Larger \( \theta \) means more coverage of the polar regions. \( F \in \{0, 1, \dots, P-1\} \) is a phasing parameter, an integer value that expresses the relative degree separation between neighboring satellites across adjacent planes, given by \( F \times 360^\circ/P\). 

Essentially, a constellation topology is represented as a graph $G(V, E)$, where $V$ is the set of satellites in the constellation and $E$ is the total number of ISLs. In practice, satellites and inter-satellite links may fail due to unpredictable factors such as space debris, solar radiation, and deliberate attacks. Therefore, at time $t$, the number of operational satellites and ISLs in the constellation are $|V(t)|\leq n$ and $|E(t)|\leq 2n$, respectively.

 \textit{Uniform traffic pattern:}  each satellite transmits data packets to all other satellites in the constellation with probability $1/(n-1)$. Compared with localized traffic patterns, this traffic model injects the maximum number of data packets into the constellation per time slot and exhibits global dependence on the connectivity of the constellation, which helps provide a deeper understanding of constellation capacity.

  %Since the probability of ISL link failure is relatively low, we assume that $0<\beta<\alpha$.
 \subsection{Memory-Aware ISL Failure Model}

 Assume time is  slotted, each ISL can be either in the OFF state or the ON state during each time slot. The state of  ISL $i$ evolves independently according to a discrete-time Markov process $\{X_{i}(t),t>0\}$, as illustrated in Fig. \ref{Fig7} \cite{ramakanth2025}. When an ISL is in the ON state ($X_{i}(t)$=1), data packets are successfully transmitted; otherwise, the transmission fails.  The steady-state probabilities of the $i$-th ISL being in states “0” and “1” are given by
\begin{equation}
\pi_1^{(i)} = \frac{\beta}{\alpha + \beta}, \quad
\pi_0^{(i)} = \frac{\alpha}{\alpha + \beta},
\label{eq:steady_state_probs}
\end{equation} 
where $\alpha$ and $\beta$ are the failure probability and the recovery probability, respectively.
\begin{figure}[!h]
\centering
\begin{tikzpicture}[>=Stealth, node distance=3cm, auto]
    % States
    \node[state] (1) {$1$}; % 连通 (On)
    \node[state, left=of 1] (0) {$0$}; % 断开 (Off)

    % Transitions
    \path[->]
    % 自环箭头放在状态上方
    (1) edge[loop above] node {$1-\alpha$} () % 保持连通
        edge[bend left=20] node {$\alpha$} (0) % 故障：1 -> 0
    (0) edge[loop above] node {$1-\beta$} () % 保持断开
        edge[bend left=20] node {$\beta$} (1); % 修复：0 -> 1
\end{tikzpicture}
\caption{Two-state discrete Markov chain  describing the ISL state evolution. State 0 corresponds to an OFF,  while state 1 corresponds
to an ON.}
\label{Fig7}
\end{figure}
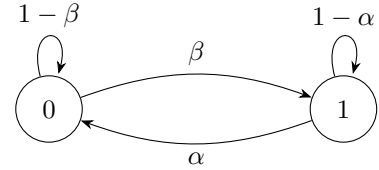

Some ISLs may experience failures, they remain in either state 0 or state 1 over multiple consecutive time slots, indicating that the evolution of ISL states exhibits memory. In other words, the state at time slot $t$ is correlated with the state at time slot $t-k$, where $t<k$. To capture the correlation, we define the belief of an ISL as the probability that the ISL is ON at time slot $t$ given its historical state at time slot $t-k$. Based on the aforementioned Markov process, the belief of ISL $i$ is given by

\begin{align}
\gamma_i(t) &= \mathbb{P}(\text{ISL } i \text{ is ON}|\text{ History state at time slot $t-k$}) \notag \\
&= \mathbb{P}\bigl(X_i(t) = 1 \bigm| X_i(t - k) = s_i\bigr), \:\:s_{i}\in \{0,1\}.
\end{align}

The $k$-step transition probabilities of the Markov process depicted in Fig. \ref{Fig7} are given by
\begin{align}
p_{00}^{(i)}(k) &= \frac{\alpha + \beta \mu^k}{\alpha + \beta}, \quad & p_{01}^{(i)}(k) &= \frac{\beta - \beta \mu^k}{\alpha + \beta}, \\
p_{10}^{(i)}(k) &= \frac{\alpha - \alpha \mu^k}{\alpha + \beta}, \quad & p_{11}^{(i)}(k) &= \frac{\beta + \alpha \mu^k}{\alpha + \beta}.
\end{align}
where $\mu=1 - \alpha - \beta$  is the \emph{memory parameter} of ISLs. 

The larger $\lvert \mu \rvert$, the higher the predictability of the ISL$_z$ state; conversely, the smaller $\lvert \mu \rvert$, the lower the predictability. When the memory parameter $\mu = 0$, the ISLs have no memory (i.i.d. in time), and when $\mu \to 1$, the ISLs remain static once realized. 
%We assume that $\mu \ge 0$, i.e., the links exhibit positive memory. 
%This assumption ensures that an ISL $i$ in the ON state is more likely to remain ON after $k$ time slots than an ISL in the OFF state turning ON after $l$ time slots, i.e.,
%\[
%p_{01}^{(i)}(l)< p_{11}^{(i)}(k), \quad \forall \:i\in E.
%\]

For analytical simplicity, we assume that the evolution of all ISLs is homogeneous.  Given a fixed $k$, a larger $\mu$ results in lower overhead for maintaining ISLs. This failure model has been used to describe the state evolution of wireless channels in opportunistic communication systems \cite{7953675}.

\section{Main Results \label{SECII}}
We make the following (nearly) minimal set of assumptions:

(A.1)  The transmission rate of each ISL is  1 Gbps.

(A.2)  In the packet multi-hop routing process, the overhead of a single-hop transmission is $\sigma <\infty$   \text{Gbps}. The sources of overhead include establishing/maintaining routing paths, sharing ISLs, queuing, and avoiding/handling congestion, etc. $\sigma$ is called as the \emph{contention overhead}.

(A.3) The overhead of establishing and maintaining each ISL is $\omega < \infty$ Gbps. This overhead depends on the behavior of ISL failures and the memory  of the ISL state evolution.  $\omega$ is called as the \emph{consensus overhead}.

(A.4)  The contention overhead and  consensus  overhead are mutually independent.

When the global constellation state information is perfectly known, all protocols and strategies incur zero overhead. Under this condition, the upper bound of the constellation capacity is achievable. Applying the bisection argument, the constellation capacity under a uniform traffic pattern is given as follows.
\begin{lemma}
Under the uniform traffic pattern, the upper bound of the constellation capacity is $8\sqrt{n}$\: Gbps. This upper bound can be achieved by employing shortest-hop routing.
\label{Th1}
\end{lemma}

\begin{IEEEproof}The proof can be found in \cite{1642730,sun2003capacity}.

\end{IEEEproof}

\begin{remark}
Apart from the shortest-hop routing, in order to achieve the capacity upper bound, the constellation topology must satisfy a symmetric geometric condition: $|M-P| < \epsilon$, where $\epsilon >0$.  In other words, the difference between the number of orbital planes and the number of satellites per plane should not be too large.
\label{RE1}
\end{remark}

From Lemma \ref{Th1},  the capacity gain  from increasing the number of satellites scales as $O(1/\sqrt{n})$. This implies that if the  overhead grows with the number of satellites at a rate greater than $\Omega(1/\sqrt{n})$, the constellation capacity will  be exhausted. In particular, in failure scenarios, the \emph{consensus overhead} associated with maintaining ISLs increases significantly. Therefore, we use the ratio of constellation capacity to total overhead as a metric to quantify the change in effective constellation capacity as the number of satellites increases.

Next, we derive a lower bound on the \emph{consensus overhead} from an information-theoretic perspective.
\subsection{Lower Bound on Consensus Overhead\label{IV_A}}

\begin{proposition}
Assuming that the ISL state evolution follows the two-state discrete-time Markov chain, the average  \emph{consensus overhead} over $k$ time slots is lower bounded by
\begin{equation}
\omega\geq\bar{H}_k(\alpha,\beta) = \frac{h\left( \pi_{1} \right) + (k-1) \left( \pi_{1} h(\alpha) + \pi_{0} h(\beta) \right)}{k},
\label{eqP1}
\end{equation}
where $h(x) = -x\log_2x - (1-x)\log_2(1-x)$ and $\pi_{0}$ and $\pi_{1}$ are given in (\ref{eq:steady_state_probs}). 

\label{P1}
\end{proposition}

\begin{IEEEproof}To maintain the constellation structure, it is essentially necessary to eliminate the state uncertainty accumulated over $k$ time slots for each ISL. Therefore, the joint entropy can be used to measure the total consensus overhead for maintaining each ISL $i$ over $k$ time slots. Applying the chain rule for entropy, we have
\begin{align}
H(X_1, X_2, \dots, X_k) &= H(X_1) + H(X_2 \mid X_1) \notag \\
&\quad + H(X_3 \mid X_1, X_2) + \dots \notag \\
&\quad + H(X_k \mid X_1, X_2, \dots, X_{k-1})
\label{eq:joint_entropy_chain_rule}
\end{align}

Using the Markov property of ISL state evolution, (\ref{eq:joint_entropy_chain_rule}) is simplified to
\begin{equation}
H(X_1, \dots, X_k) = H(X_1) + (k-1)H(X_2|X_1).
\label{eq6}
\end{equation}

By incorporating the steady-state probabilities of ISL state evolution (see (\ref{eq:steady_state_probs})), we obtain
\begin{align}
H(X_2 \mid X_1) 
&= \sum_{j \in \{0,1\}} \pi_j H(X_2 \mid X_1 = j) \notag \\
&= \pi_1 h(\alpha) + \pi_0 h(\beta),
\label{eq:cond_entropy}
\end{align}
where $h(x) = -x\log_2x - (1-x)\log_2(1-x),\: x>0$.

The prior uncertainty $H(X_{1})$ associated with each ISL is given by
\begin{equation}
H(X_1) = h\bigl(\pi_1\bigr).
\label{eq:prior_entropy}
\end{equation}

Based on equations (\ref{eq6}), (\ref{eq:cond_entropy}) and (\ref{eq:prior_entropy}), the average consensus overhead is lower bounded by $$\omega\geq \frac{h\left( \pi_{1} \right) + (k-1) \left( \pi_{1} h(\alpha) + \pi_{0} h(\beta) \right)}{k}.$$
\end{IEEEproof}

Proposition \ref{P1} derives an information-theoretic lower bound on the \emph{consensus overhead} required for maintaining each ISL. This bound is jointly determined by the link-state evolution behavior and the maintenance period $k$. Next, we investigate the impact of different parameters on $\bar{H}_k(\alpha,\beta)$.

\begin{corollary}
Given the link-state evolution dynamics, i.e., fixed $\alpha$  and $\beta$, the following results hold: 
\begin{itemize}
    \item When $k = 1$, $\omega \geq  h\left( \pi_{1} \right)$.
    \item When $k \to +\infty$, $$\omega \geq \bar{H}_{\infty}(\alpha, \beta) =  \frac{\beta}{\alpha+\beta} h(\alpha) + \frac{\alpha}{\alpha+\beta} h(\beta)>0.$$
\end{itemize}
where $\pi_{1}=\frac{\beta}{\alpha+\beta}$.
\label{Coro1}
\end{corollary}
\begin{IEEEproof}
Two cases are considered separately: setting \(k = 1\) and  \(k \to +\infty\) in (\ref{eqP1}). This completes the proof of Corollary \ref{Coro1}.

\end{IEEEproof}

When $k=1$, the ISL is maintained in each time slot, without exploiting the prior state information of the ISL. In this case, the minimum consensus overhead equals the entropy of ISL  in the “on” state. In contrast, as $k \to +\infty$, the fixed overhead associated with the initial confirmation of the ISL state becomes negligible. In this case, the \emph{consensus overhead} depends solely on the ISL state changes frequently.
Furthermore, it can be  observed that when the ISL experiences failures,  we have 
\begin{equation}
h\left( \pi_{1} \right)\geq\bar{H}_{\infty}(\alpha, \beta), 
\end{equation}
where $0<\alpha+\beta\leq 1$. The equality holds if and only if $\alpha + \beta = 1$.

Increasing the ISL maintenance period $k$ can reduce the \emph{consensus overhead}, but it cannot fall below $\bar{H}_{\infty}(\alpha, \beta)$.  Fig.~\ref{fig:2} shows the lower bound of the consensus overhead as a function of the failure probability $\alpha$ for $\beta = 0.25, 0.5,$ and $0.75$.  When the failure-recovery probability $\beta > 0.5$ and the \emph{memory parameter} of the ISL $\mu=1-\alpha-\beta > 0$, the consensus overhead required under the $k=1$ and $k \to +\infty$ strategies becomes nearly identical. In other words, maintaining the ISL in each time slot (i.e., $k=1$) is the optimal strategy, since $h\left( \pi_{1} \right) \approx \bar{H}_{\infty}(\alpha, \beta)$. In contrast, when $\alpha + \beta > 1$, increasing the maintenance period $k$ can reduce the \emph{consensus overhead} compared with $k=1$, since the predictability of the ISL state increases, even though the ISL state changes frequently.
 \begin{figure*}[!t]
    \centering
    % 子图(a)
    \begin{subfigure}[b]{0.32\textwidth}
        \centering
        \includegraphics[width=\textwidth]{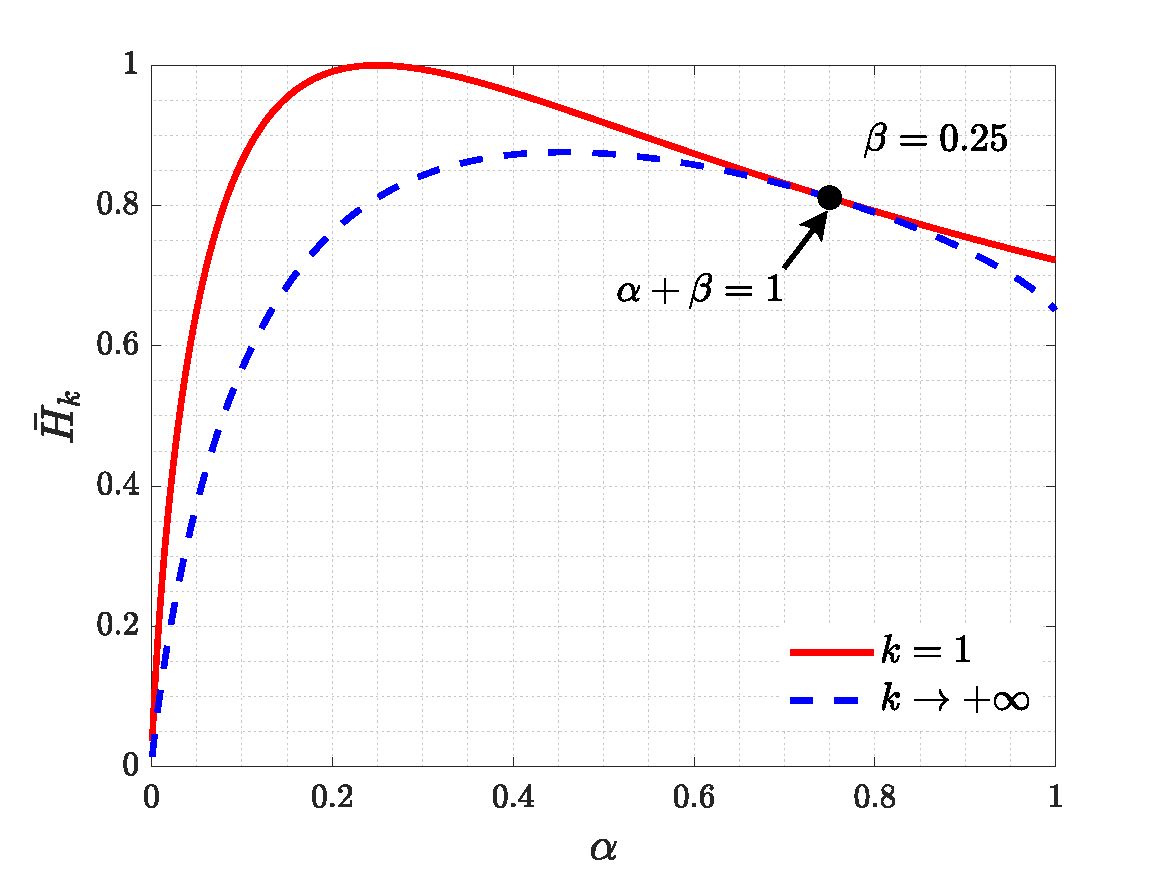}
        \caption{$\beta = 0.25$}
        \label{fig:sub1}
    \end{subfigure}
    \hfill
    % 子图(b)
    \begin{subfigure}[b]{0.32\textwidth}
        \centering
        \includegraphics[width=\textwidth]{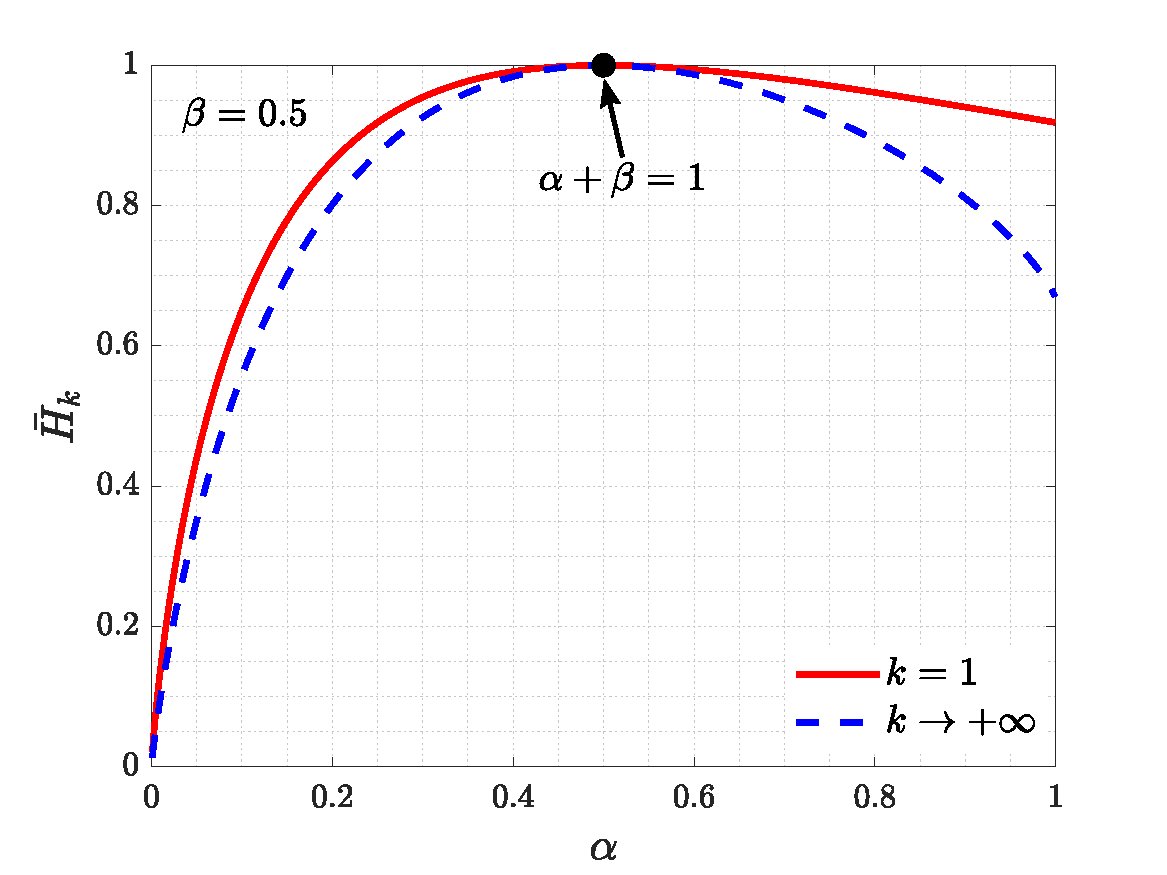}
        \caption{$\beta = 0.5$}
        \label{fig:sub2}
    \end{subfigure}
    \hfill
    % 子图(c)
    \begin{subfigure}[b]{0.32\textwidth}
        \centering
        \includegraphics[width=\textwidth]{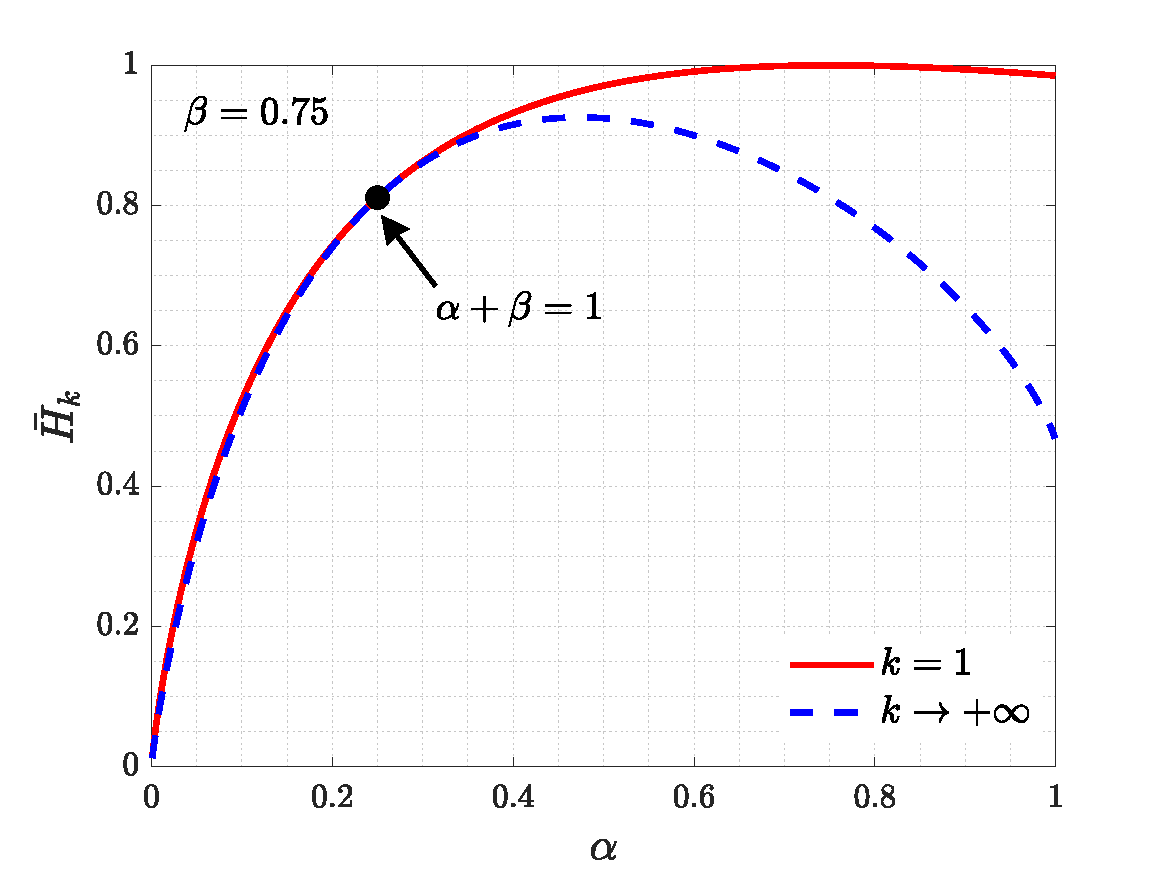}
        \caption{$\beta = 0.75$}
        \label{fig:sub3}
    \end{subfigure}

    \caption{Lower bound of \emph{consensus overhead} as a function of failure probability $\alpha$ under different recovery probabilities $\beta$. }
    \label{fig:2}
\end{figure*}

When $k = c < +\infty$, the relative magnitudes of the failure probability $\alpha$ and the recovery probability $\beta$ affect the lower bound of the \emph{consensus overhead}. The evolutionary behavior of ISLs is classified into five regions. Next, we focus on analyzing the lower bounds of consensus overhead in regions I–IV, as shown
in Fig.\ref{Cor2_Fig2}.  The evolutionary behavior of ISL in each region is as follows:

\emph{Region I}: The failure probability and recovery probability of ISL are both far below 1. The probability of ISL transitioning from the ON (OFF) state to the OFF (ON) state is low, indicating that the state of ISL is stable.

\emph{Region II}: The failure probability of ISL is low, while the recovery probability is high. The probability of ISL transitioning from the  OFF state to the  ON state is  high.

\emph{Region III}: Both the failure probability and recovery probability of ISL are high. The state of ISL transitions frequently.

\emph{Region IV}: The failure probability of ISL is high, while the recovery probability is low. Once a failure occurs, the ISL is difficult to recover.
\begin{figure}[!htbp]
    \centering
    \includegraphics[scale=0.7]{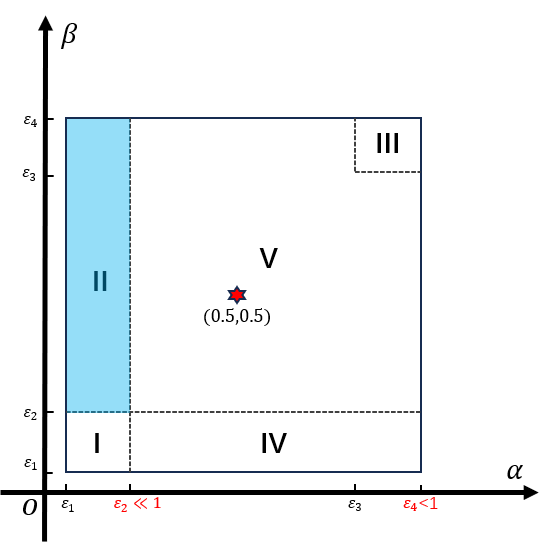}
    \caption{Classification of ISL state evolution behaviors for $0<\epsilon_{1}<\epsilon_{2}\ll\epsilon_{3}<\epsilon_{4}<1$. The lower bound of the \emph{consensus overhead} is maximized when $\alpha = \beta = 0.5$.}
    \label{Cor2_Fig2}
\end{figure}
%\vspace{-0.2cm}

Each region can be represented as:
\[
\begin{aligned}
\mathcal{A}_{1} &= \{ (\alpha, \beta) \mid \epsilon_{1} \leq \alpha \leq \epsilon_{2},\ \epsilon_{1} \leq \beta \leq \epsilon_{2} \},\\[2mm]
\mathcal{A}_{2} &= \{ (\alpha, \beta) \mid \epsilon_{1} \leq \alpha \leq \epsilon_{2},\ \epsilon_{2} < \beta \leq \epsilon_{4} \},\\[1mm]
\mathcal{A}_{3} &= \{ (\alpha, \beta) \mid \epsilon_{3} \leq \alpha \leq \epsilon_{4},\ \epsilon_{3}\leq \beta \leq \epsilon_{4} \},\\[1mm]
\mathcal{A}_{4} &= \{ (\alpha, \beta) \mid \epsilon_{2} < \alpha \leq \epsilon_{4},\ \epsilon_{1} < \beta < \epsilon_{2} \},\\[1mm]
\mathcal{A}_{5} &= \mathcal{A} - \bigcup_{i=1}^{4} \mathcal{A}_{i},
\end{aligned}
\]
where $\mathcal{A}=\{ (\alpha, \beta) \mid 0<\epsilon_{1} \leq \alpha,\beta \leq \epsilon_{4} <1\}$.

\begin{corollary}
Given the maintenance period of ISLs , i.e., $k=c>1 $, the following results hold: 
\\

\centering
\small
\renewcommand{\arraystretch}{1}
\begin{tabularx}{\columnwidth}{@{} >{\centering\arraybackslash}X *{3}{>{\centering\arraybackslash}X} @{}}
\toprule
\textbf{Region ID} & \textbf{I \& III} & \textbf{II} & \textbf{IV} \\ 
\midrule
$\bar{H}_{k}$ & $\frac{1}{k}$ & $\frac{k-1}{k} h(\alpha)$ & $\frac{k-1}{k} h(\beta)$ \\
\bottomrule
\end{tabularx}
\label{Coro2}
\end{corollary}
\begin{IEEEproof} When $\alpha \approx \beta$, we have $h(\pi_1) \approx 1$. Based on this fact, we derive the lower bound on the \emph{consensus overhead} for each region.

 For Region I, when $\alpha \approx \beta \ll 1$, we obtain
\begin{equation}
h(\alpha) \approx 0 \quad \text{and} \quad h(\beta) \approx 0.
\label{C2_eq11}
\end{equation}

Substituting (\ref{C2_eq11}) into (\ref{eqP1}) yields 
$\omega \geq \frac{1}{k}. $

Using the same approach, the lower bound on the \emph{consensus overhead} for Region III is also $\frac{1}{k}$.

 For Region II, when $0<\alpha \ll 1$, we have
 \begin{equation}
\pi_{0} \approx 0 \quad \text{and} \quad  h(\beta) \approx 0.
\label{C2_eq12}
\end{equation}

Substituting (\ref{C2_eq11}) into (\ref{eqP1}) yields 
$\omega \geq \frac{k-1}{k} h(\alpha). $

By using the similar approach as in Regions  II, we can complete the proof for Regions  IV.

\end{IEEEproof}

%\begin{figure*}[!t]
%    \centering
%    % 子图(a)
%    \begin{subfigure}[b]{0.32\textwidth}
%        \centering
%        \includegraphics[width=\textwidth]{C2_K_5}
%        \caption{$k=5$}
%        \label{fig:sub1}
%    \end{subfigure}
%    \hfill
%    % 子图(b)
%    \begin{subfigure}[b]{0.32\textwidth}
%        \centering
%        \includegraphics[width=\textwidth]{C2_K_10}
%        \caption{$k=10$}
%        \label{fig:sub2}
%    \end{subfigure}
%    \hfill
%    % 子图(c)
%    \begin{subfigure}[b]{0.32\textwidth}
%        \centering
%        \includegraphics[width=\textwidth]{C2_K_5-10}
%        \caption{$\bar{H}_{5}-\bar{H}_{10}$}
%        \label{fig:sub3}
%    \end{subfigure}
%
%    \caption{Lower bound on the \emph{consensus overhead} versus $\alpha$ and $\beta$ for different fixed $k$.}
%    \label{fig:4}
%\end{figure*}

It can be observed that when the ISL state evolution lies in Regions I and III, selecting a larger maintenance period $k$ can reduce the \emph{consensus overhead}. In other words, the ISL states in these regions exhibit strong memory, i.e., the memory parameter $|\mu|$ is large. In contrast, when the ISL state evolution lies in Region II, increasing $k$ can hardly reduce the \emph{consensus overhead}. In this region, the recovery probability  is relatively high and the \emph{consensus overhead} is mainly incurred to address the ISL uncertainty caused by  failures. Similarly, when the ISL state evolution lies in Region IV, increasing $k$ yields only limited reductions in \emph{consensus overhead}. In this case, once an ISL failure occurs, the recovery probability   is very low, and the overhead is mainly incurred to eliminate the uncertainty regarding whether the ISL has recovered. 

% To provide further insight, Figs.~\ref{fig:4}a and \ref{fig:4}b illustrate the lower bound of the \emph{consensus overhead} with respect to the ISL evolution parameters for $k=5$ and $k=10$, respectively. Fig.~\ref{fig:4}c shows the reduction in overhead when $k$ is increased from 5 to 10.

\subsection{ Upper Bound  of Capacity Scalability}
We integrate the overhead into the constellation capacity analysis and use it to quantify how ISL state evolution and routing affect the capacity scalability. From the overhead perspective, the capacity scalability   is defined as follows.

\begin{definition}
The capacity scalability is defined as the ratio of the constellation capacity under non-failure conditions to the total overhead, which is denoted by $\tau$.
\label{DES}
\end{definition}

Using Lemma \ref{Th1} and Proposition \ref{P1}, the upper bound of the capacity scalability  is derived.

\begin{theorem}
Under the uniform traffic pattern, if the ISLs follow the two-state discrete Markov chain and the maintenance period is $k\geq1$, the upper bound of capacity scalability is given by
\begin{equation}
   \tau \leq \tau(n,k)=\frac{16\sqrt{n}}{1+ \sigma n^{1.5} + 4n\bar{H}_k(\alpha,\beta)},
    \label{TH1}
\end{equation}
where $\alpha$ and $\beta$ are the failure probability and the recovery probability, respectively.  When the total overhead converges to 0,  (\ref{TH1}) degenerates to Lemma \ref{Th1}.
\label{STH1}
\end{theorem}

\begin{IEEEproof} From Lemma \ref{Th1}, it can be seen that under ideal conditions, the upper bound of the constellation capacity is $8\sqrt{n}$. The total overhead of the constellation is denoted by 
\begin{equation}
\phi(n) = \phi_1(n) + \phi_2(n),
\label{TH1_eq00} 
\end{equation}
where $\phi_1(n)$ is the contention overhead and $\phi_2(n)$ is the consensus overhead.

 \emph{Contention Overhead $\phi_1(n)$}:  we first prove the average hop count under the uniform traffic pattern.  For any satellites $i$ and $j$ with $j \neq i$, the minimum hop count between them is given as follows:
\begin{equation}
d(i,j)
= \min\{\delta_x, \sqrt{n} - \delta_x\} + \min\{\delta_y, \sqrt{n} - \delta y\},
\label{TH1_eq14} 
\end{equation}
where $\delta_x = |x_i - x_j|$ and $\delta_y = |y_i - y_j|$.

Combining (\ref{TH1_eq14}), the average hop count is given as follows
\begin{equation}
\overline{L}(n) = \frac{1}{n} \sum_{i \in V} \sum_{j \in V} \text{Pr}(i, j) d(i,j).
\label{TH1_eq15} 
\end{equation}
where the probability $\text{Pr}(i, j)$ that satellite $i$ is matched with satellite $j$ to form a source--destination pair is 
\[
\text{Pr}(i, j) = 
\begin{cases}
\frac{1}{n - 1}, &  i \neq j, \\
0, & i = j.
\end{cases}
\]

Substituting (\ref{TH1_eq14}) into (\ref{TH1_eq15}), we obtain
\begin{equation}
\overline{L}(n) = 
\begin{cases}
\frac{\sqrt{n}}{2}, & \sqrt{n}\:\: \text{odd} \\
\frac{n\sqrt{n}}{2(n-1)}, & \sqrt{n}\:\: \text{even}
\end{cases}
= \frac{\sqrt{n}}{2} + O(1/\sqrt{n}).
\label{TH1_eq16} 
\end{equation}

According to the construction rules of the uniform traffic pattern, the total number of source--destination pairs in the constellation is $n$. 
The total number of hops required to successfully deliver all packets generated by the source satellites to their corresponding destination satellites is $n\overline{L}(n)$.

 Combining Assumption (A.2), the total \emph{contention overhead} is
\begin{equation}
\phi_{1}(n)\geq \sigma n\overline{L}(n)\approx \frac{\sigma n^{1.5}}{2}.
\label{TH1_eq18} 
\end{equation}

\emph{Consensus Overhead $\phi_2(n)$}:  For a constellation of size $n$, the total number of ISLs is at most $2n$. Combining Proposition \ref{P1}, we have
\begin{equation}
\phi_{2}(n)= 2n\omega\geq2n\bar{H}_k(\alpha,\beta).
\label{TH1_eq19} 
\end{equation}

Substituting (\ref{TH1_eq18}) and (\ref{TH1_eq19}) into (\ref{TH1_eq00}), we obtain
\begin{equation}
\phi(n)\geq \frac{\sigma n^{1.5}}{2}+2n\bar{H}_k(\alpha,\beta).
\end{equation}

Based on Definition \ref{DES}, Theorem \ref{STH1} is proved. 

\end{IEEEproof} 

\begin{remark}
When the total contention overhead is constant ($\sigma = c$), the upper bound of capacity scalability increases from $O(1/n)$ to $O(1/\sqrt{n})$. This merely reduces the convergence rate of capacity scalability to zero.
\end{remark}

When the constellation size grows, the total overhead increases much faster than the constellation capacity, i.e., when ISLs are unstable, the constellation does not have scalability.
This shows that building large constellations depends on reducing overhead rather than on designing the constellation structure.  This is because, for any topology, the upper bound of capacity  is $O(\sqrt{n})$. Routing strategies with low overhead should be used, and the link maintenance period $k$ should be chosen properly according to the ISL evolution behavior. If this is not possible, the frequency of using multi-hop transmission should be reduced.

To ensure that the capacity of the constellation for carrying data packets remains at $\Theta(\sqrt{n})$, 
the contention overhead and the consensus overhead should satisfy the following conditions: $\sigma = o(n^{-1.5})$ and $\omega = o(n^{-1})$, where $\epsilon$ is an arbitrarily small positive real number. This condition is extremely stringent. As the constellation size increases, the capacity available for ISL maintenance and routing becomes lower.

By fitting the simulation data from small-scale constellations, the competitive overhead ($\sigma$) and consensus overhead ($\omega$) in (\ref{TH1}) can be obtained. This enables the prediction of the scalability of constellations using specific routing algorithms and ISL maintenance strategies. This is particularly important for the design of network protocols in large-scale constellations, especially under limited computational resources for simulations.

\subsection{ Maximum Capacity Scalability }

Based on the above analysis, for a given routing algorithm and ISL maintenance strategy, there exists an optimal constellation size $n_{*}$, at which the capacity scalability  reaches its maximum. Let $ Q(n,k,\alpha,\beta)\triangleq \frac{16\sqrt{n}}{1+ \sigma n^{1.5} + 4n\bar{H}_k(\alpha,\beta)}$.   By solving 
\begin{equation}
\frac{\partial Q}{\partial n}=0,
\label{TH1_eq22}
\end{equation}
the optimal constellation size $n_{*}$ is obtained.

\begin{theorem}
Given the ISL evolution behavior and the routing algorithm, the optimal constellation size $n_{*}$ satisfies
\begin{equation}
2\sigma n_{*}^{1.5} + 4n_{*}\bar{H}_k(\alpha,\beta)  =1
\label{eq:optimal_n}
\end{equation}
The maximum scalability  is $8\sqrt{n_{*}}$.
\label{STH2}
\end{theorem}

\begin{IEEEproof}By taking the partial derivative,  (\ref{TH1_eq22}) is simplified to
\begin{equation}
 \frac{8\left(1 - 2\sigma n^{1.5} - 4\bar{H}_k(\alpha,\beta) n\right)}{\sqrt{n}\left(1 + \sigma n^{1.5} + 4n\bar{H}_k(\alpha,\beta)\right)^2}=0.
\label{TH1_eq23}
\end{equation}

Since $\sqrt{n}\left(1 + \sigma n^{1.5} + 4n\bar{H}_k(\alpha,\beta)\right)^2 > 0$, $n_{*}$ satisfies 
\begin{equation}
2\sigma n^{1.5} + 4n\bar{H}_k(\alpha,\beta)  =1
\label{TH1_eq24}
\end{equation}

Substituting (\ref{TH1_eq24}) into $Q(n,k,\alpha,\beta)$ yields the maximum capacity scalability. This is because $\frac{\partial^2 Q}{\partial n^2} < 0$, i.e.,

\[
\frac{\partial^2 Q}{\partial n^2} \Bigg|_{n = n_*}
= \frac{-8\left(3\sigma\sqrt{n_*} + 4\bar{H}_k(\alpha,\beta)\right)}
       {\sqrt{n_*}\left(1 + \sigma(n_*)^{1.5} + 4 n_* \bar{H}_k(\alpha,\beta)\right)^2} < 0.
\]
\end{IEEEproof}

The capacity scalability  first increases and then decreases with the number of satellites. 
This implies that, for a given routing algorithm and link maintenance strategy, there exists an optimal number of satellites that maximizes capacity scalability.  This maximum is extremely sensitive to the state of the ISLs. Specifically, when the recovery probability is given, increasing the failure probability reduces the maximum capacity scalability  $\tau$ and shifts the optimal number of deployed satellites to the left. 
A numerical example is shown in Fig.~\ref{TH2_Fig5}.

\begin{figure}[!htbp]
    \centering
    \includegraphics[scale=0.6]{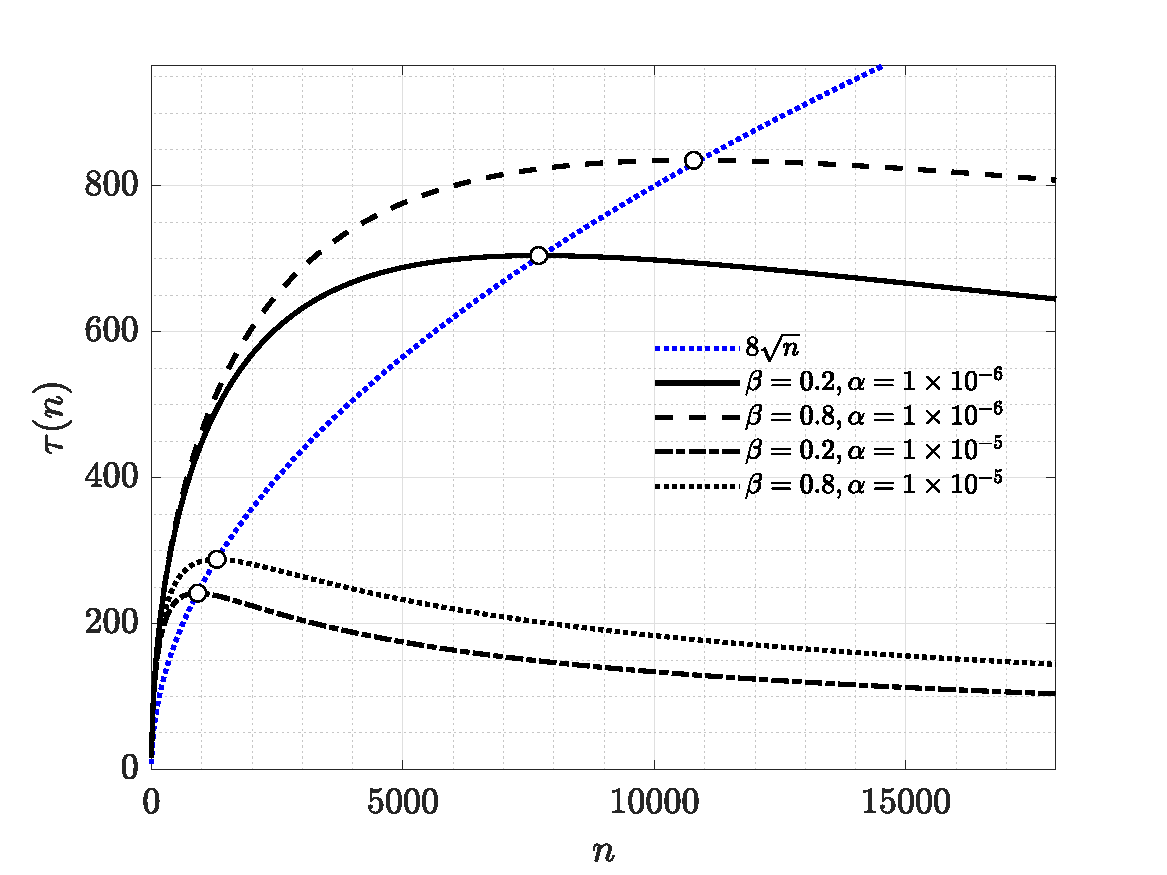}
    \caption{Scalability  vs. $n$ under different ISL evolution behaviors.}
    \label{TH2_Fig5}
\end{figure}

Corollary~\ref{Coro_TH2} provides the optimal $n_{*}$ in two extreme scenarios: when $\sigma \to 0$ and when $\bar{H}_k(\alpha,\beta) \to 0$.

\begin{corollary} The optimal constellation sizes under the two extreme scenarios are given as follows:

Scenario I: When the contention overhead is zero, i.e., $\sigma \to 0$,
\begin{equation}
n_{*}^{1} \approx \frac{1}{\bar{H}_k(\alpha,\beta)}.
\label{C3_eq26}
\end{equation}

Scenario II: When the coherency overhead is zero, i.e., $\bar{H}_k(\alpha,\beta) \to 0$,

\begin{equation}
n_{*}^{2} \approx \left( \frac{1}{2\sigma} \right)^{2/3}.
\label{C3_eq27}
\end{equation}
\label{Coro_TH2}
\end{corollary}
\begin{IEEEproof} 
This proof is completed by analyzing equation (\ref{eq:optimal_n}).

\end{IEEEproof}

It can be easily observed that, for the same reduction in overhead, the optimal constellation size is more sensitive to consensus overhead than to contention overhead. This implies that reducing the ISL maintenance consensus overhead yields the greatest improvement in capacity scalability. In Table \ref{table_identical_initial}, when the overheads in both scenarios are reduced by 10 times, the optimal constellation size in Scenario I increases by 10 times, which is far larger than the 4.6 times increase observed in Scenario II.

\begin{table}[!htbp]
\caption{Optimal size $n_*$ under 10$\times$ overhead reduction.}
\label{table_identical_initial}
\centering
\footnotesize
\setlength{\tabcolsep}{4pt} 
\renewcommand{\arraystretch}{1.3} 
\begin{tabular}{ccccc}
\toprule
\textbf{Scenario} & \textbf{Overhead} & \textbf{$n_*$ (Initial)} & \textbf{$n_*$ (Optimized)} & \textbf{Gain} \\
\midrule
I 
& \multirow{2}{*}{$10^{-6} \to 10^{-7}$} 
& $2.5\times 10^5$ 
& $2.5\times 10^6$ 
& 10 times \\[1.5pt] 

II 
% 第二列已由 multirow 覆盖，不再写
& 
& $6.3\times 10^3$ 
& $2.9\times 10^4$ 
& 4.6 times \\ 
\bottomrule
\end{tabular}
\end{table}
\section{Further Discussion \label{SECV}}

As shown in Section~\ref{SECII}-B, the ISL evolution behavior and the maintenance period affect both the consensus overhead and the optimal constellation size. In this section, a detailed analysis is presented. 

\subsection{Capacity Scalability vs. Maintenance Period}
From Corollary~\ref{Coro1}, the choice of the  maintenance period ($k$) influences the lower bound of the consensus overhead under a given ISL evolution behavior. In this subsection, we analyze the impact of the maintenance period on scalability. Based on Corollary~\ref{Coro1}, the following conclusion can be drawn.

\begin{corollary}When the ISL evolution behavior is given (i.e., $\alpha$ and $\beta$ are fixed), for any $k \in [1,+\infty)$, the upper bound of \textit{capacity scalability} satisfies the following inequality: 
\begin{equation}
     \tau_{1}(n) \leq \tau(n, k) < \tau_{\infty}(n),\label{C4_eq28}
\end{equation}
where 
    \begin{equation}
        \tau_{1}(n) = \frac{16\sqrt{n}}{1 + \sigma n^{1.5} + 4n  h\left( \pi_{1}\right)}
        \label{C4_eq29}
    \end{equation}
    and
    \begin{equation}
        \tau_{\infty}(n) = \frac{16\sqrt{n}}{1 + \sigma n^{1.5} + 4n \left[ \pi_{1} h(\alpha) + \pi_{0} h(\beta) \right]}
    \end{equation}
$\pi_{1}$ and $\pi_{0}$ are given in (\ref{eq:steady_state_probs}).
    \label{Coro4}
    \end{corollary}
    
\begin{IEEEproof} Substituting the lower bound of the consensus overhead in Corollary~\ref{Coro1} into (\ref{TH1}), the proof is completed.

\end{IEEEproof} 

 Increasing the maintenance period $k$ cannot change the inherent trend that scalability first increases and then decreases with $n$. This is because the denominator terms in both $\tau_{1}(n)$ and $\tau_{\infty}(n)$ still follow  $f(n)=1+an^{1.5}+cn$, where $a,c>0$. Although increasing $k$ cannot change the fact that scalability converges to zero as $n$ increases, it can shift the optimal constellation size to a larger value, at which the scalability  is maximized. Note that this gain does not continue to increase with $k$. In practice, it is sufficient to choose a relatively large maintenance period. A numerical example is given in Fig. \ref{C4_fig6}.  The maximum gain of scalability with respect to $n$ also follows the pattern of first increasing and then decreasing (see Fig. \ref{fig6:sub2}). When $n \to \infty$, we have 
\begin{equation}
\lim_{n \to \infty} \Delta \tau(n) = 0,
\label{Delta1}
\end{equation}
where $\Delta \tau(n) =\tau_{\infty}(n)-\tau_{1}(n)$. This implies that, as the constellation size increases, the scalability gain obtained by increasing $k$ converges to 0.
  
\begin{figure}[!htbp]  % 单栏环境
    \centering
    % 子图1
    \begin{subfigure}[b]{\linewidth}  % 占满整栏
        \centering
        \includegraphics[width=.9\linewidth]{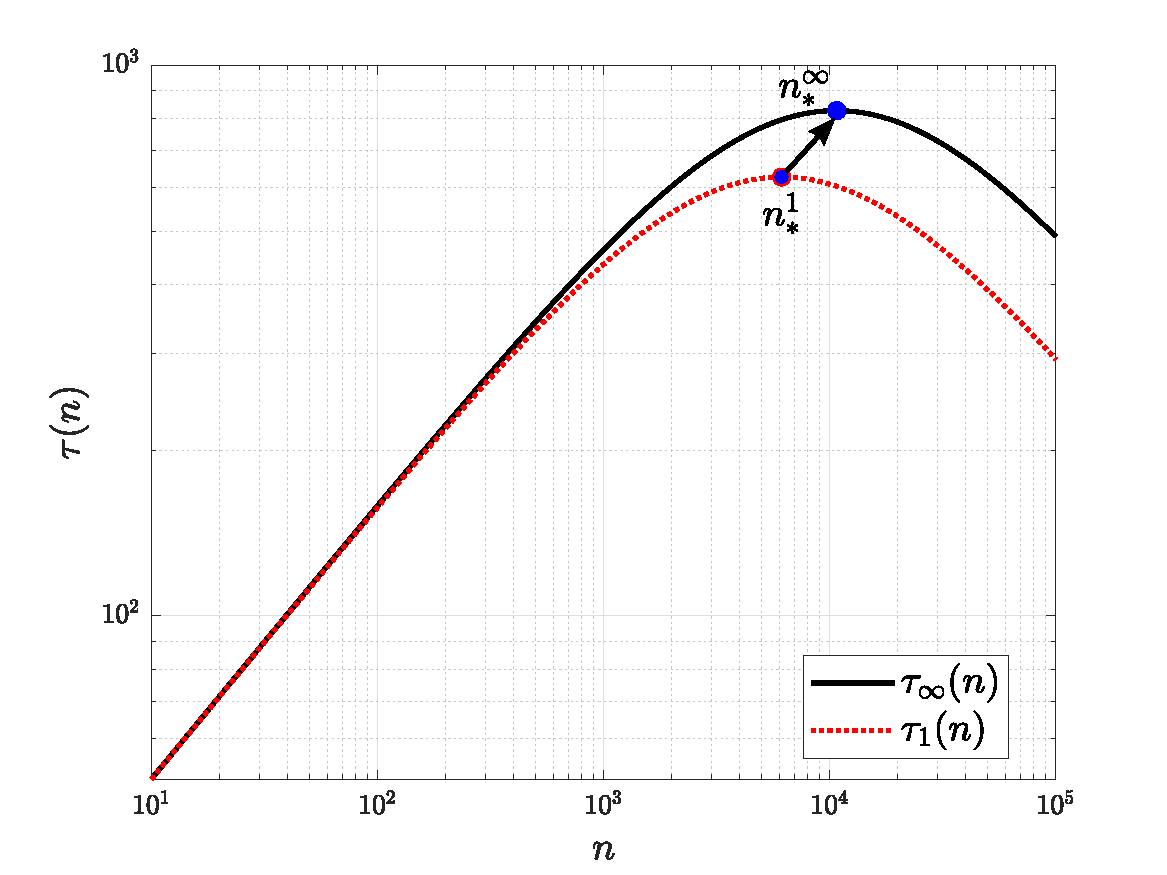}  % 图像宽度可调
        \caption{$\tau(n)$ vs. $n$}
        \label{fig6:sub1}
    \end{subfigure}
    \begin{subfigure}[b]{\linewidth}  % 占满整栏
        \centering
        \includegraphics[width=.9\linewidth]{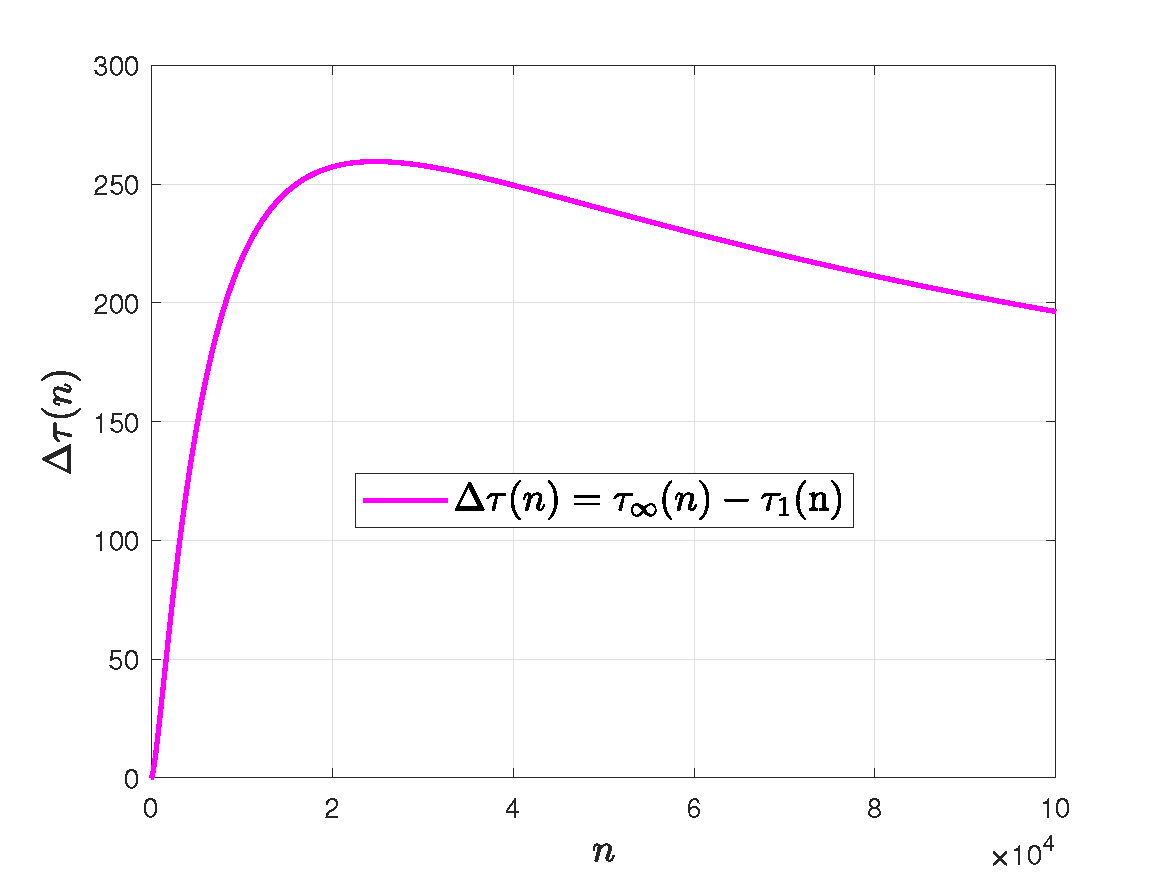}
        \caption{$\Delta \tau(n)$ vs. $n$}
        \label{fig6:sub2}
    \end{subfigure}
 \caption{Scalability versus $n$ for maintenance periods $k=1$ and $k \to \infty$.  The contention overhead, failure probability and recovery probability are $\sigma =  10^{-10}$, $\alpha =  10^{-6}$ and $\beta = 0.5$, respectively. }
    \label{C4_fig6}
\end{figure}
    
\subsection{Capacity Scalability vs. ISL  Evolution Behavior}
In practice, ISL states vary with the space environment, resulting in differing consensus overheads. The lower bound on the consensus overhead varies across different regions (see Corollary \ref{Coro2}). Based on Corollary \ref{Coro2}, we derive the upper bounds on the scalability of \emph{Regions} I-IV.  This result provides insight into the variation of the scalability upper bound when ISLs undergo extreme evolution behavior.

\begin{corollary}
When the maintenance period is given, i.e., $k = c \geq 1$, the upper bounds on the capacity scalability in \emph{Regions} I-IV are as follows:

\vspace{5pt} 
\centering
\small
\renewcommand{\arraystretch}{1.2} % 降低全局行高，使标题行变窄
\begin{tabularx}{\columnwidth}{@{} >{\centering\arraybackslash\bfseries}m{2cm} >{\centering\arraybackslash}X @{}}
\toprule
Region ID & \bfseries Upper Bound  \\ 
\midrule
% 为涉及分式的行单独增加垂直间距，确保公式不拥挤
\addlinespace[5pt]
I\&III & $\dfrac{16k\sqrt{n}}{k(1 + \sigma n^{1.5}) + 4n}$ \\[5pt] 
\midrule
\addlinespace[5pt]
II & $\dfrac{16k\sqrt{n}}{k(1 + \sigma n^{1.5}) + 4n(k-1)h(\alpha)}$ \\[5pt]
\midrule
\addlinespace[5pt]
IV & $\dfrac{16k\sqrt{n}}{k(1 + \sigma n^{1.5}) + 4n(k-1)h(\beta)}$ \\[5pt]
\bottomrule
\end{tabularx}
\vspace{5pt}

\end{corollary}
\begin{IEEEproof} Substituting the lower bound of the consensus overhead in Corollary~\ref{Coro2} into (\ref{TH1}), the proof is completed.

\end{IEEEproof} 

The evolution behaviors of ISLs in Regions I and III exhibit a stark contrast in their microscopic dynamics. In Region I, both the failure probability and the recovery probability are low. This means that once an ISL fails, it remains in the OFF state for a long time. In contrast, in Region III, both the failure and recovery probabilities are high, so the ISL remains in the ON state for only a short time. Interestingly, despite these  differences,   the ISL state evolution is highly predictable in Regions I and III, as the memory parameter $|\mu| = |1 - \alpha - \beta| \to 1$. The capacity scalability of Regions I and III is identical in the information-theoretic sense.  In practice, increasing ISL link resilience and reducing the number of  satellites should be used to prevent ISL evolution behavior from entering Regions I and III. Moreover, when ISL evolution behavior is in Regions I and III, increasing the maintenance period is not a good strategy because the ISL state memory parameter $\mu \to 1$.

  However, in Region II, the failure probability is low and the recovery probability is high, meaning that ISLs rarely fail. In this case, the scalability is optimal because $h(\alpha) \to 0$.  Compared to Region II, in Region IV, the failure probability is high and the recovery probability is low, making it nearly impossible for ISLs to recover after failure. As the number of satellites increases, the number of failed ISLs in the constellation increases because $h(\beta) \to 0$. Therefore, maintaining ISL evolution behavior in Region II is the optimal  for achieving the highly capacity scalability. A numerical example is shown in Fig. \ref{Cor5_fig1}. When $k=10$, using the same routing strategy --- $\sigma=10^{-11}$, the scalability in different ISL evolution regions satisfies $\tau_{2}(n)>\tau_{4}(n)>\tau_{1}(n)=\tau_{3}(n)$.
  
  \begin{figure}[!htbp]
    \centering
    \includegraphics[scale=0.625]{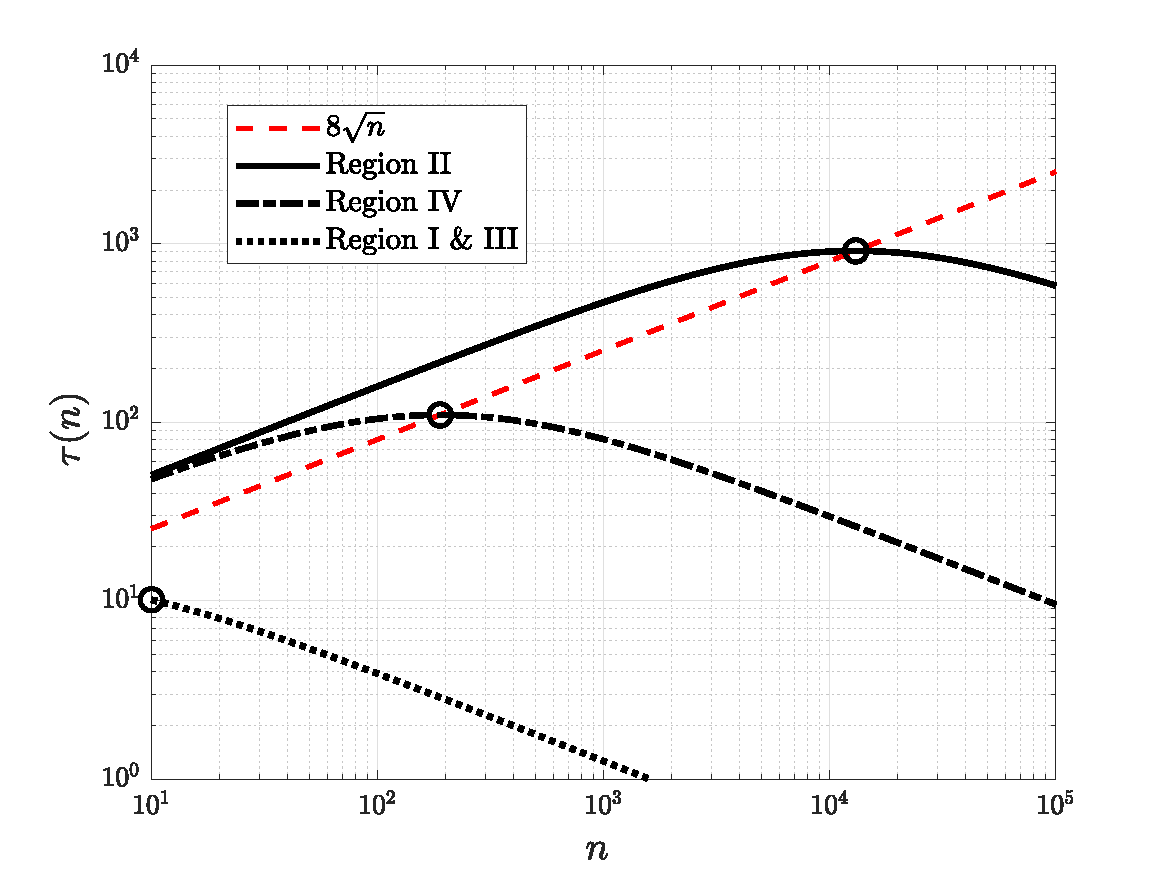}
    \caption{Capacity Scalability  vs. $n$ under different  regions.  $\epsilon_1 = 10^{-6},\;\epsilon_2 = 10^{-4},\;\epsilon_3 = 0.7,\;\epsilon_4 = 0.9$.}
    \label{Cor5_fig1}
\end{figure}

\section{Simulation Examples: Capacity Scalability with Shortest-hop Routing\label{SECVI}}

Assuming that the shortest-hop routing strategy is employed, we will simulate and verify the capacity scalability of the constellation under different regions of ISL evolution behavior. 

\subsection{Simulation of the Evolutionary Behavior of Constellation Structures}

To facilitate better understanding, the  evolution of the  number of ISLs in the ON state  and the constellation connectivity under different  parameters is presented in Fig. \ref{fig:four_sub8}. The  connectivity is defined as the ratio of the number of satellites in the largest connected component to the total number of satellites, denoted as $\xi(t)$.  The number of ISLs in the ON state characterizes local behavior, while connectivity reflects the global behavior of the constellation topology.

In Region I, both $\alpha$ and $\beta$ are very small, making the ISL states difficult to change. As a result, the constellation connectivity may suddenly collapse or  recover after remaining at a certain level for a period of time (see Fig. \ref{fig:sub1}). In Region II, the repair probability is high and the failure probability is low. $E_{\text{on}}(t)$ stays at a high level. The constellation connectivity is 1 with high probability (see Fig. \ref{fig:sub2}).  Region II is also the optimal operating region. In Region III, both the failure probability and the repair probability are high. This means that the ISL states change much more frequently compared to Region I. The constellation connectivity is extremely unstable and changes almost every time slot. Therefore, when the ISL states evolve in Region III, protocol parameters must be adjusted in every time slot.   Moreover, when the ISLs evolve in Region V, the parameters are set as $\alpha = \beta = 0.5$ because the entropy of each ISL reaches its maximum. In this case, the ISL states are almost impossible to predict because both $E_{\text{on}}(t)$ and $\xi(t)$ fluctuate violently (see Fig.\ref{fig:sub4}). 

In the simulation examples of Regions I, III and V, the average number of ISLs in the ON state is 10000, but the variances are significantly different. Comparing Figs. \ref{fig:sub2} and \ref{fig:sub3}, although the number of ISLs in the ON state fluctuates greatly in both cases, the connectivity behavior shows significant differences. This indicates that the number of ISLs in the ON state determines the constellation connectivity.
\begin{figure*}[!t]
    \centering
    % 第一行子图
    \begin{subfigure}[b]{0.45\textwidth}
        \centering
        \includegraphics[width=\linewidth]{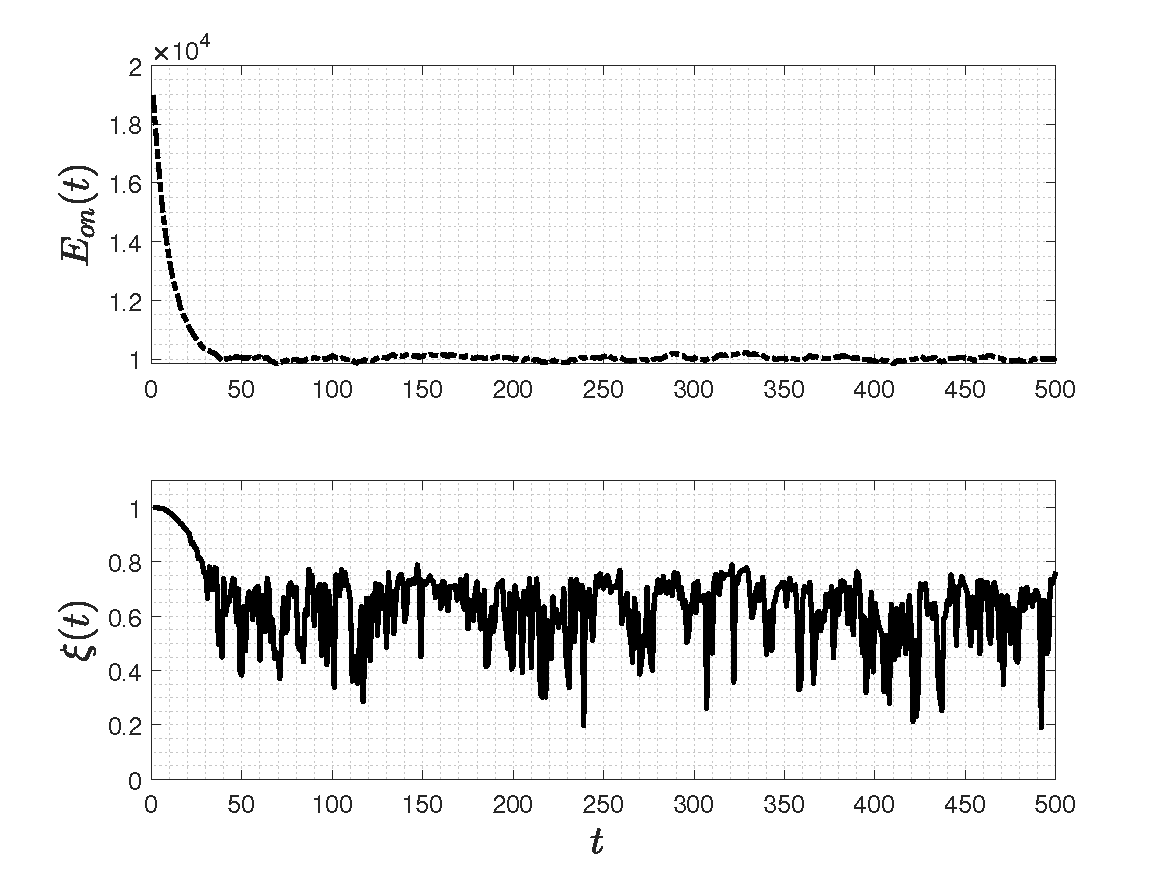}
        \caption{Region I: $\beta=\alpha=.05$}
        \label{fig:sub1}
    \end{subfigure}
    \hfill
    \begin{subfigure}[b]{0.45\textwidth}
        \centering
        \includegraphics[width=\linewidth]{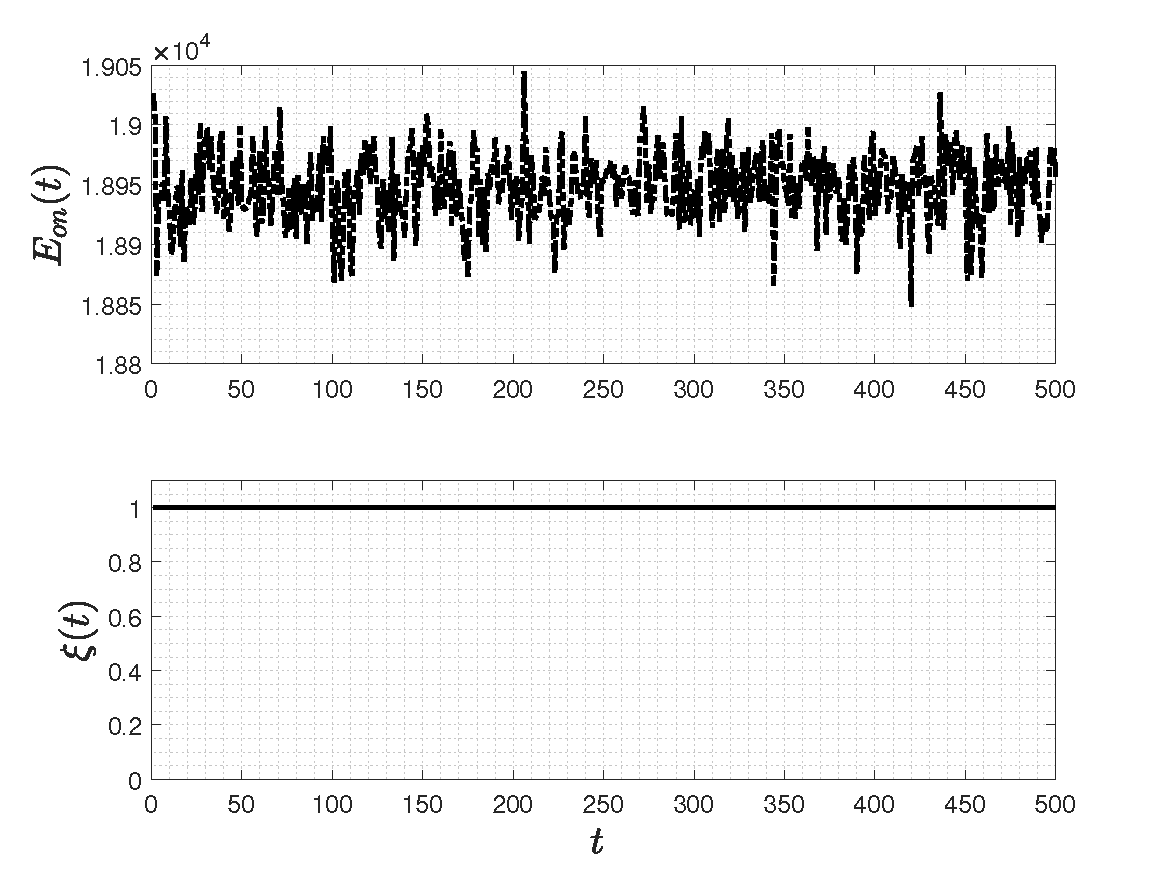}
        \caption{Region II: $\beta=.9$,\:$\alpha=.05$}
        \label{fig:sub2}
    \end{subfigure}

    \vspace{0.1cm} % 行间距，可调整

    % 第二行子图
       \begin{subfigure}[b]{0.45\textwidth}
        \centering
        \includegraphics[width=\linewidth]{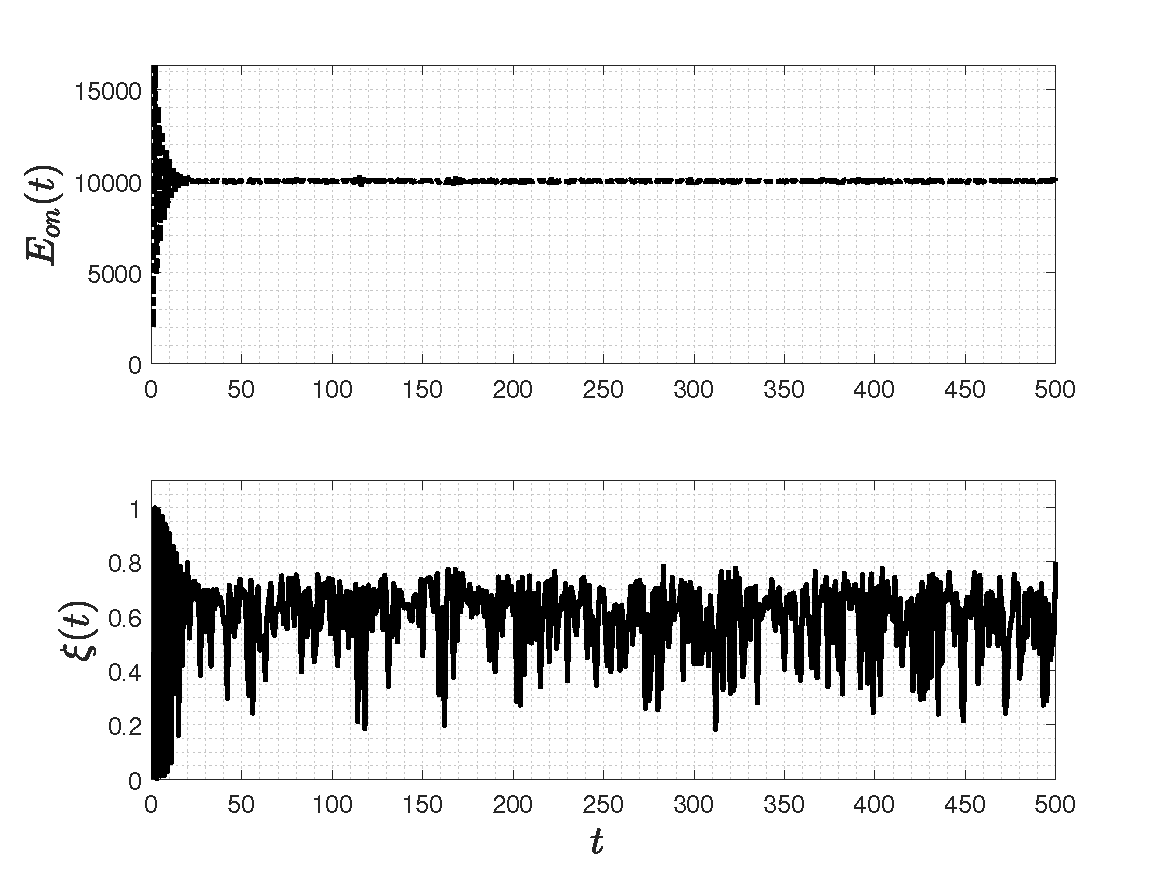}
        \caption{Region III: $\beta=\alpha=.9$ }
        \label{fig:sub4}
    \end{subfigure}
    \hfill
        \begin{subfigure}[b]{0.45\textwidth}
        \centering
        \includegraphics[width=\linewidth]{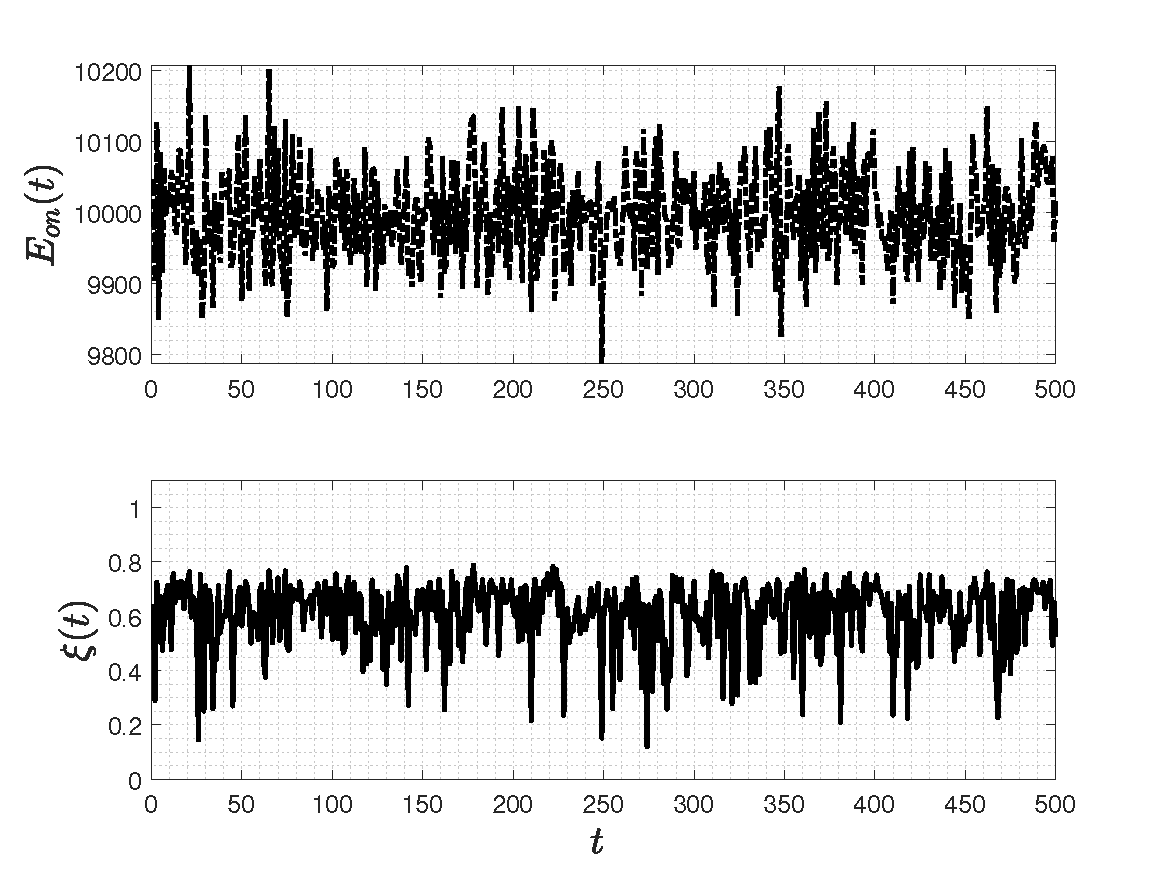}
        \caption{Region V: $\beta=\alpha=.5$}
        \label{fig:sub3}
    \end{subfigure}

    \caption{Evolution trajectories of the number of ISLs in the ON state and the constellation connectivity when the ISLs are located in different regions. The total number of satellites is 10,000.}
    \label{fig:four_sub8}
\end{figure*}

\subsection{Simulation of Capacity Scalability}
The simulation procedure for capacity scalability under dynamic link failures is as follows: For a constellation of size $n$, $n$ source-destination pairs are generated through  random permutation. The constellation topology is constructed, where each satellite initially establishes four ISLs. Each ISL evolves for $t_{0}$ time slots according to the given failure and recovery probabilities until reaching a steady state. Starting from time slot $t_{0} + 1$, the shortest hop count for each source-destination pair is computed using Dijkstra's algorithm. If no path is found, it is recorded as $0$. This process is repeated for $t_{1}$ time slots. The state transition probabilities of the ISLs and the average hop count are then evaluated and substituted into $\tau(n,k)$ in (\ref{TH1}) to obtain the simulation results. In computing the shortest paths between source–destination pairs, we use the perfect constellation topology information. The shortest hop count for successfully establishing a source–destination routing path is minimized.

Given $\sigma =  10^{-12}$, $t_{0} = 100$, $t_{1} = 800$, and $k = 10$, we set $(\alpha, \beta) = ( 10^{-5}, 0.8) \in \text{Region II}$. The capacity scalability is simulated for $n = 100, 400, 900, 1600, 2500, 3600, 6400$ and $8100$. Each data point is obtained by averaging over 1000 simulation runs. The simulation results are shown in Fig. \ref{fig9}. Under perfect constellation topology information, the simulated capacity scalability can attain the theoretical upper bound. This implies that shortest-hop routing is one of the strategies for achieving the upper bound of capacity scalability, with consensus overhead minimized. The capacity scalability is extremely sensitive to fluctuations in the evolutionary state of ISLs. When capacity scalability reaches its maximum, the simulation results exhibit greater variability. The  overhead for maintaining the constellation structure and discovering routing paths becomes unstable. Under imperfect constellation topology information, the capacity scalability is further reduced.  This is because the overhead required for establishing routing paths and maintaining ISLs increases with the scale of the satellite constellation. If the constellation structure is used to support local traffic rather than all-to-all traffic, capacity scalability can be further improved because the minimum hop count between source–destination pairs is reduced.

From Theorem \ref{STH1}, the functional form of capacity scalability under the all-to-all traffic model is given as follows:
$$f(n)=\frac{16\sqrt{n}}{1+an\bar{L}+4bn\bar{H}_{k}(\alpha,\beta)},$$
$a$ and $b$ are determined by the routing strategy and the ISL maintenance protocol, respectively. For any given set of protocols, the average hop count $\bar{L}$ and the state transition probabilities of the ISLs can be obtained through simulations on a small scale constellation. The parameters $\hat{a}$ and $\hat{b}$ are then fitted from the simulation data. Based on the empirical formula, we predict whether the constellation scalability reaches its maximum. For example, under minimal consensus overhead, the shortest-hop routing strategy yields $\hat{a} = \mathrm{5.44}e^{-7}$ and $\hat{b} = \mathrm{0.965}$.
  \begin{figure}[!htbp]
    \centering
    \includegraphics[scale=0.625]{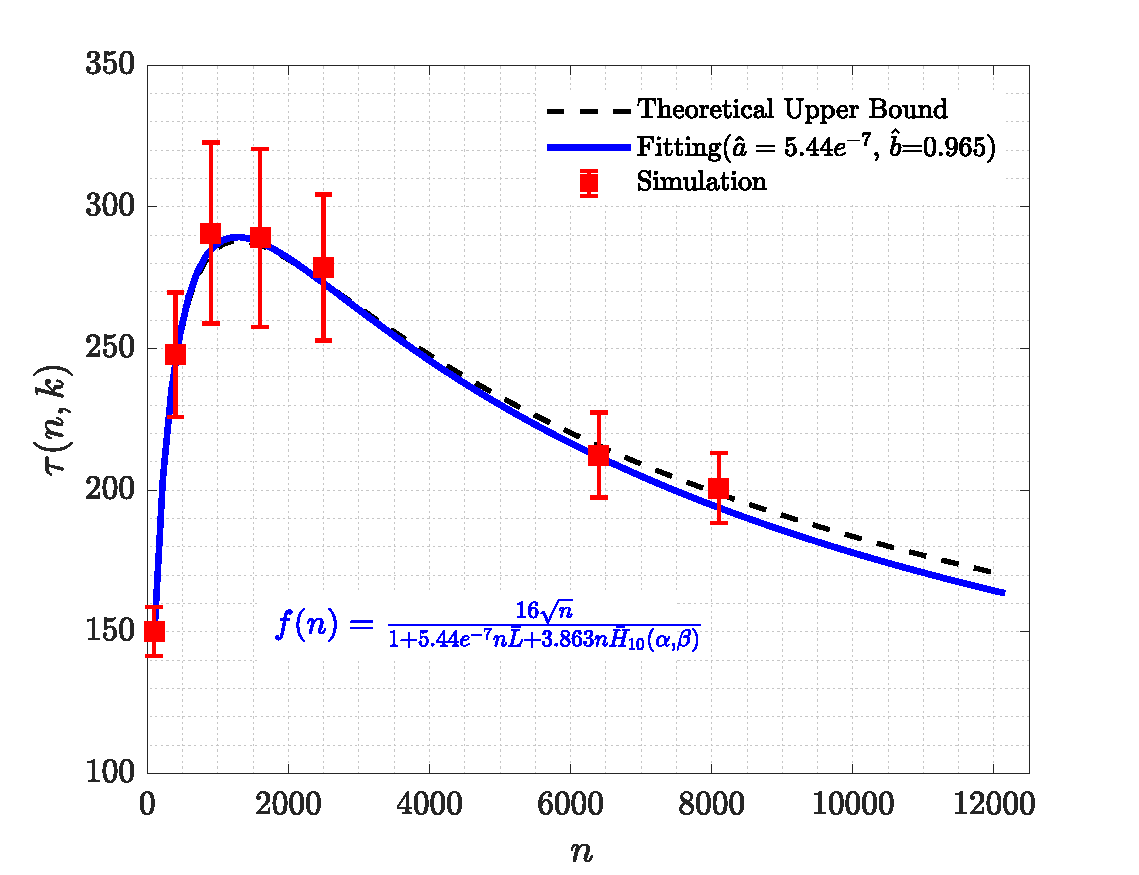}
    \caption{Capacity scalability  vs. $n$.  }
    \label{fig9}
\end{figure}
\section{CONCLUSIONS\label{SECVII}}
The upper bound on capacity scalability under uniform traffic patterns is derived by incorporating the failure evolution behavior of ISLs. For any protocol, the capacity scalability converges to zero as the constellation size approaches infinity. The optimal constellation size is identified. When the constellation size is below this optimal value, the capacity scalability increases with constellation size, implying that scaling up the network can achieve higher effective capacity for data packet transport. In the future, quantifying capacity scalability under non-uniform traffic patterns remains an important research direction. It can be anticipated that capacity scalability under non-uniform traffic is higher than that under uniform traffic, while the asymptotic convergence to zero still holds. This is because, under dynamic failures, the consensus overhead  scales as $\Theta(n)$.

\bibliographystyle{IEEEtran}
\bibliography{Capacity_vs_Overhead}

\end{document}